\documentclass[12pt]{amsart}  
\usepackage{amscd}
\usepackage{amsmath,amssymb,amsthm,yfonts,mathrsfs,mathtools}

\newtheorem{prop}{Proposition}
\newtheorem{lem}{Lemma}
\newtheorem{cor}{Corollary}

\theoremstyle{definition}

\theoremstyle{definition}

\newcommand{\te}{\theta_e}

\newcommand{\mA}{\mathcal{A}}
\newcommand{\bZ}{\mathbf{Z}}
\newcommand{\bC}{\mathbf{C}}
\newcommand{\mT}{\mathcal{T}}

\numberwithin{equation}{section}

\begin{document}

\title[Algebraic integrability of the classical XXZ spin chain]{Algebraic integrability of the classical XXZ spin chain with reflecting boundary conditions.  }
\author{Gus Schrader}
\address{Department of Mathematics, University of California, Berkeley, CA
94720, USA}
\maketitle

\begin{abstract}
In this paper we analyze the classical XXZ spin chain with reflecting boundaries.  We exhibit a system of log-canonical coordinates on the phase space generalizing Sklyanin's separation of variables for the periodic XXZ chain, and use these coordinates to construct action-angle variables for the system.  We also integrate the flows of the reflection Hamiltonians explicitly in terms of Riemann theta functions.  Central to our analysis is the algebraic integrability of the model.  
\end{abstract}

\section{Introduction. }
A family of integrable reflecting boundary conditions for the quantum XXZ spin chain was introduced by Sklyanin in \cite{sklyaninref}.  In our recent work \cite{gus}, we studied the Poisson geometry of the classical analogues of the quantum integrable systems constructed by Sklyanin.  We showed that such systems fit into a general framework of integrable systems on Poisson symmetric spaces $G/K$, where $(G,r)$ is a quasitriangular Poisson-Lie group and the subgroup $K$ is the fixed point set of a solution $\sigma$ of the classical reflection equation.  In this construction, the Hamiltonians of the classical XXZ chain with reflecting boundaries arise as elements of a certain Poisson commutative subalgebra of (twisted) bi-invariant functions on the loop group $LSL_2$.  

\bigskip
The goal of the present paper is to analyze the particular example of the classical XXZ spin chain with reflecting boundaries in greater detail.  As noted in \cite{sklyaninsov}, there are three fundamental problems in the analysis of integrable Hamiltonian systems.  They are:
\begin{itemize}
\item Separation of variables on the system's phase space
\item Integration of the system's equations of motion
\item Construction of the system's action-angle variables
\end{itemize}
In this work we address all three problems.  Crucial to our analysis is the so-called {\em algebraic integrability} of the system: the tori on which the flows of the reflection Hamiltonians are linearized are in fact abelian varieties, arising as Jacobians of the spectral curves of the reflection monodromy matrix.  

\bigskip

The plan of the paper is as follows.  In section two, we introduce the phase space of the model, and recall how to construct the Poisson commutative subalgebra of reflection Hamiltonians using the reflection monodromy matrix. We also write down the equations of motion generated by reflection Hamiltonians, which are shown to take the Lax form. 

 In the third section, we study the various spectral curves associated to the model, the holomorphic differentials on them, and the morphisms between them.
 
In the section four, we write down a system of log-canonical coordinates on the phase space generalizing Sklyanin's separation of variables for the periodic XXZ chain.  In the fifth section, we use these coordinates and the Hamilton-Jacobi method to integrate flow of the reflection Hamiltonians by quadratures. This construction also reveals the algebraic integrability of the system.  Section six contains the construction of the complex action-angle variables. Finally, in section seven we exploit the geometric description of the reflection flows to write formulas for the time evolution of the reflection monodromy matrix in terms of Riemann theta functions.

\section{The integrable system.}
In this section we recall the construction of the classical XXZ spin chain with reflecting boundaries, as described in \cite{gus} and references therein.  To each site of the chain, we associate a copy of the $SL_2^*$, the Poisson Lie group dual to $SL_2$ with its standard Poisson-Lie structure. Explicitly, $\bC[SL_2^*]=\bC[e,f,k^{\pm1}]$, with the Poisson bracket given by
\begin{align}
\label{epb}
&\nonumber\{k,e\}= ke\\ & \{k,f\}=- kf \\ \nonumber & \{e,f\}=2(k^{2}-k^{-2})
\end{align}
The function $\omega=k^2+k^{-2}+ef$ is a Casimir element of the Poisson algebra $\bC[SL_2^*]$, and its generic level set $\Sigma_t=\{\omega=t\}$ is a two-dimensional symplectic leaf in the Poisson manifold $SL_2^*$. 

\bigskip

It is convenient to gather the generators $e,f,k$ into the $2\times2$ matrix Laurent polynomial
$$
L(z)=\left(\begin{array}{cc} zk-z^{-1}k^{-1} & e\\f&zk^{-1}-z^{-1}k  \end{array}\right)
$$
The Poisson brackets (\ref{epb}) can be recast in terms of $L(z)$ with the help of the classical trigonometric $r$-matrix
\begin{align}
\label{rmatrix}
r(z_1/z_2)&=\frac{1}{2(z_1^2-z_2^2)}\left[\begin{array}{cccc}z_1^2+z_2^2&0&0&0\\
						0&-(z_1^2+z_2^2)&4z_1z_2&0\\
						0&4z_1z_2&-(z_1^2+z_2^2)&0\\
						0&0&0&z_1^2+z_2^2
 \end{array}\right]\end{align}
We then have
\begin{align}
\label{rtt}
\{L_1(z_1),L_2(z_2)\}=[r_{12}(z_1/z_2),L_1(z_1)L_2(z_2)]
\end{align}
The $\bC[SL_2^*]\otimes\bC[z,z^{-1}]$-valued matrix $L(z)$ also satisfies
\begin{align}
\label{det}
\det L(z)=z^2+z^{-2}-\omega
\end{align}
as well as the identities
\begin{align}
\label{identities}
L(z)L(z^{-1})&=-\det L(z)\mathrm{Id}\\
L(z^{-1})^t&=-\sigma_2L(z)\sigma_2^{-1}\\
L(-z)&=-\sigma_3L(z)\sigma_3^{-1}
\end{align}
where 
$$
\sigma_2=\left(\begin{array}{cc}0&-1\\1&0\end{array}\right), \ \ \sigma_3=\left(\begin{array}{cc}1&0\\0&-1\end{array}\right)
$$
are the Pauli matrices.  

\bigskip

The phase space of the $N$-site spin chain is the $2N$-dimensional symplectic manifold
$$
M_{2N}=\Sigma_{t_1}\times\cdots\times\Sigma_{t_N}
$$ 
To write down the Hamiltonians of the integrable spin chain, we first fix the data of a diagonal solution of the reflection equation
\begin{equation*}
K(z)=\left(\begin{array}{cc}{\xi z-z^{-1}\xi^{-1}} &0\\0&{\xi z^{-1}-z\xi^{-1}} \end{array}\right)
\end{equation*}
together with an $N$-tuple of non-zero complex numbers $a_1,\ldots,a_N$.  We may then form the {\em reflection monodromy matrix}
\begin{align}\mT(z)=\left(\frac{1}{z-z^{-1}}\right)L_1(a_1z)&\cdots L_N(a_Nz)K(z)L_N(z/a_1)\cdots L_1(z/a_N)\\
&=:\left(\begin{array}{cc}A(z) & B(z)\\ C(z)& D(z)\end{array}\right)
\end{align}
From the symmetries (\ref{identities}) of $L(z)$, it follows that $\mT(z)$ satisfies
\begin{align}
\mT(-z)&=\sigma_3\mT(z)\sigma_3^{-1}\\
\mT(z^{-1})^t&=-\sigma_2\mT(z)\sigma_2^{-1}
\end{align}
From the formula (\ref{rtt}), one finds that matrix elements of $\mT(z)$ have the Poisson brackets
\begin{align}
\label{refpb}
\nonumber\{\mT_1(z_1)\otimes\mT_2(z_2)\}=[r_{12}&(z_1/z_2),\mT_1(z_1)\mT_2(z_2)]\\
& + \mT_1(z_1)r_{12}(z_1z_2)\mT_2(z_2)-\mT_2(z_2)r_{12}(z_1z_2)\mT_1(z_1)
\end{align}
Later, we will need the following explicit formulae for the Poisson brackets of matrix elements of $\mT(z)$:
\begin{align}
\label{exppb}
\{A(z_1),A(z_2)\}&= \frac{2}{z_1z_2-z_1^{-1}z_2^{-1}}\bigg(B(z_1)C(z_2)-C(z_1)B(z_2) \bigg)  \\
\{C(z_1),A(z_2)\}&= \frac{2z_1}{(z_2^2-z_1^2)(z_1^2z_2^2-1)}\bigg( z_1z_2^4C(z_1)A(z_2)-z_1^2z_2^3A(z_1)C(z_2)-\\ \nonumber&z_1^2z_2D(z_1)C(z_2)+z_2^3D(z_1)C(z_2)-z_1C(z_1)A(z_2)+z_2A(z_1)C(z_2)\bigg)
\end{align}
Note that unlike in the periodic case, the functions $A(z)$ do not form a Poisson commutative family.  

\bigskip

  The {\em reflection transfer matrix} is the Laurent polynomial $t(z)$ defined by
 $$
 t(z)=\frac{1}{2}\mathrm{tr} \ \mT(z)
 $$
 \begin{prop}
{\em (\cite{sklyaninref},\cite{gus})} The reflection transfer matrix satisfies $$\{t(z_1),t(z_2)\}=0$$
 and thus its coefficients generate a Poisson commutative subalgebra in $\bC[M_N]$.  
 \end{prop}

   Let us describe some properties of the transfer matrix. Firstly, by the symmetries (\ref{identities}) of $\mT(z)$, we have
 $$
t(-z)=t(z), \ t(z^{-1})=-t(z)   
 $$
The transfer matrix $t(z)$ therefore a function of the variable $w=z^2$, which admits an expansion 
\begin{align}
t(z)=\frac{1}{2}\left(\frac{w+1}{w-1}\right)\left(P_{N}\left(\frac{w^N+w^{-N}}{2}\right)+P_{N-1}\left(\frac{w^{N-1}+w^{1-N}}{2}\right)+\cdots+P_0 \right)
\end{align}
Note also that
$$
t(z)=\frac{A(z)-A(z^{-1})}{2}
$$
with the function $A(z)$ taking the form
$$
A(z)=\frac{Pz^{2N+1}+\cdots-P^{-1}z^{-2 N-1}}{z-z^{-1}}
$$
where the leading coefficient
$$
P=\xi_-\prod_{j=1}^Nk_j^2
$$ 
is proportional to the deformed total $\sigma^z$-component of spin. The leading coefficient of the transfer matrix is
\begin{align*}
\frac{P_{N}}{2}&=\xi_-\prod_{j=1}^N(k_j)^2-\xi_-^{-1}\prod_{j=1}^N(k_j)^{-2}\\
&=P-P^{-1}
\end{align*}
The following lemma, giving a linear relation between the reflection Hamiltonians, is a simple consequence of formulas (\ref{det}) and (\ref{identities}).   
\begin{lem}
The reflection transfer matrix $t(z)$ satisfies 
$$
\lim_{z\rightarrow1} \ (z-z^{-1})t(z)=\sum_{j=0}^{N}P_{j}=\left({\xi_--\xi_-^{-1}}\right)\prod_{k=1}^N(\omega_k-a_k^2-a_k^{-2})
$$
\end{lem}
We also have the following proposition, which shows that the functions $(P_{1},\ldots,P_{N})$ form a set of $N$ functionally independent Hamiltonians.  
\begin{prop}
For generic values of the constants $\xi,\omega_i,a_i$, the reflection Hamiltonians $P_{1},\ldots,P_{N}$ are functionally independent.  
\end{prop}
\begin{proof}
Since the functional independence is an open condition, it suffices to consider the case $\xi=a_i\equiv1$.  We will prove the stronger statement that the $P_{1},\ldots,P_{N}$ remain independent when restricted to the $N$-dimensional subvariety of phase space cut out by $\{f_j=0|j=1,\ldots,N\}$. On this locus, the local Lax operators become upper triangular, so the reflection monodromy matrix becomes 
$$
t(z)=\left(\frac{z+z^{-1}}{z-z^{-1}}\right)\left(\prod_{j=1}^N(zk_j-z^{-1}k_j^{-1})^2+\prod_{j=1}^N(zk_j^{-1}-z^{-1}k_j)^2\right)
$$
Note that the reflection Hamiltonians $P_j$ are functions of the variables $\tilde{k}_j=k_j^2$: explicitly, for $1\leq j\leq N$, we have
$$
P_j=\sum_{r_i\in\{0,\pm1\},r_1+\cdots+r_k=j}\bigg(\prod_{i=1}^N(-2)^{\delta_{r_i,0}}(\tilde{k}_i^{r_i} +\tilde{k}_i^{-r_i} )\bigg)
$$  To verify their algebraic independence, it suffices to check that the Jacobian $J(\tilde{k})=\det\left[ \frac{\partial{P_i}}{\partial_{\tilde{k}_j}}\right]$ is not identically zero.  Indeed, by counting degrees one sees that the Laurent monomial $k_1^{N-1}k_2^{N-2}\cdots k_{N-1}$ can only be obtained from the diagonal term in the expansion of the determinant $J(\tilde{k})$, where it appears with coefficient $(-2)^{N(N+1)/2}$.  

\end{proof}
This proposition shows that the classical XXZ spin chain with reflecting boundary conditions is an integrable system. 
 Note that the reflection Hamiltonians can be written
$$
P_k=2^{2-\delta_{k,0}}\mathrm{Res}_{z=0}\left(   \frac{z-z^{-1}}{z+z^{-1}}\right)z^{-2k-1}t(z)dz
$$ 
Let us now write down the equations of motion generated by the $P_k$.  For this we need to introduce some notations.  Given any Laurent polynomial $$f(z)=\sum_{n\in\bZ}a_nz^n\in\bC[z,z^{-1}]$$ we may uniquely decompose $f$ as
$$
f=f^\sigma+ f^+
$$
where $f_\sigma$ satisfies $f^\sigma(z)=f^\sigma(z^{-1})$ and
 $f^+\in z\bC[z]$.  Let us also introduce the matrices
 $$
M^\sigma_k(z) =2^{2-\delta_{k,0}}\left(\left(   \frac{z-z^{-1}}{z+z^{-1}}\right)z^{-2k}\mT(z)\right)^\sigma
 $$
 $$
 M^+_k(z)=2^{1-\delta_{k,0}}\left(\left(   \frac{z-z^{-1}}{z+z^{-1}}\right)z^{-2k}\mT(z)\right)^+
 $$
Taking the trace over the first space in equation (\ref{refpb}), we find that the equations of motion take the following Lax form:
\begin{align}
\label{eom}
\frac{\partial}{\partial t_k}\mT(z):=\{\mT(z),P_k\}&=\left[ \mT(z),M^\sigma_k(z)\right]\\
&=\left[M^+_k(z), \mT(z)\right]
\end{align}
We therefore obtain the following corollary, which opens the door to studying the system using the algebro-geometric techniques explained in \cite{belokolos},\cite{babelon},\cite{harnad},\cite{RSTS} and references therein.  
\begin{cor}
\label{isospectral}
The spectrum of the reflection monodromy matrix $\mT(z)$ is preserved under the Hamiltonian flows of the reflection Hamiltonians.  In particular, the coefficients of the characteristic polynomial $\det\left(\zeta-\mT(z)\right)$ are invariant under these flows.  
\end{cor}  
%

\section{Spectral curves. }

Motivated by Corollary \ref{isospectral}, we consider the invariant spectral curve \begin{align}
\label{zcurve}
M \ : \ \det\left(\zeta-\mT(z)\right)=0  
\end{align}
cut out of $\bC\times\bC^*$ by the characteristic polynomial of the reflection monodromy matrix $\mT(z)$.  More precisely, we shall work with the compact Riemann surface obtained by adding four points at infinity, two points over $z=0$ and another two points over $z=\infty$.  In what follows, we will use the notation $M$ to refer to this compact Riemann surface.  
Introducing 
\begin{align}
\label{ycoord}
y=\zeta-t(z)
\end{align}
we have
$$
y^2=t(z)^2-\det \mT(z)
$$
 By (\ref{det}), the coefficients of $\det \mT(z)$ are constant on a symplectic leaf, so that all degrees of freedom for the moduli of $M$ are in fact encoded by the transfer matrix $t(z)$ and its coefficients $\{P_j\}$.
 
 \bigskip
 
Let us introduce the notations $\lambda=z^2+z^{-2}$ and
$$
\mathcal{Q}(z)=t(z)^2-\det \mT(z)
$$
\begin{lem}
We have $\mathcal{Q}(z)=\mathcal{Q}_{2N}(\lambda)$ where $\mathcal{Q}_{2N}(\lambda)$ is a polynomial of degree $2N$ in $\lambda$.   
\end{lem}
This fact has the following geometric meaning.  Firstly, the curve $M$ is a 4-fold cover of a genus $N-1$ hyperelliptic curve
$$
\Gamma \ : \ y^2-\mathcal{Q}_{2N}(\lambda)=0
$$
 and a 2-fold cover of the intermediate genus $2N-1$ spectral curve
$$
\Sigma \ : \ y^2-\tilde{\mathcal{Q}}(w)=0
$$ 
where $w=z^2$ and $\mathcal{Q}(z)=\tilde{\mathcal{Q}}(w)=\mathcal{Q}_{2N}(\lambda)$. The projection $\pi:\Sigma\rightarrow \Gamma$ is given by $\lambda=w+w^{-1}$.  Note that $\Gamma=\Sigma/\tau$, where $\tau:\Sigma\rightarrow\Sigma$ is the involution $\tau(w,y)=(w^{-1},y)$. 

\bigskip

We now turn to the description of the holomorphic differentials on the various spectral curves.  The space $H^0(\Sigma,K)$ of holomorphic differentials on $\Sigma$ has dimension $g(\Sigma)=2N-1$.  We may decompose $H^0(\Sigma,K)$ into its $\pm1$ eigenspaces $V_\pm$ with respect to the induced action of the involution $\tau$.  Bases may be chosen as
$$
V_+=\text{span}\left\{\omega^+_j=\frac{(w-w^{-1})(w^j+w^{-j})}{y w}dw \   \bigg| \ 0\leq j\leq N-2\right\}
$$
$$
V_-=\text{span}\left\{\omega^-_k=\frac{(w^k+w^{-k})}{yw}dw \ \bigg|  \ 0\leq k\leq N-1\right\}
$$
The subspace $V_+$ coincides with $\pi^*H^0(\Gamma,K)$, and its elements may be regarded as holomorphic differentials on $\Gamma$. The following basis for $V_+$ will prove well adapted to the description of the flows of our chosen basis of reflection Hamiltonians $P_j$:
$$
\Omega_j=\left(\frac{w+1}{w-1}\right)\frac{(w^j+w^{-j}-2)}{8yw}dw \ , \ 1\leq j\leq N-1
$$
We will also need the following differential of the third kind
$$
\Omega_{N}=-(P+P^{-1})\left(\frac{w+1}{w-1}\right)\frac{(w^{N}+w^{-N}-2)}{2yw}dw \  
$$
which has simple poles at the two points $\infty_\pm$ lying over $\lambda=\infty$ and is regular elsewhere.  We shall label the points $\infty_\pm$ by
$$
(\lambda^{-N}y)(\infty_\pm)= \pm\left(\frac{P+P^{-1}}{2}\right)
$$
Observe that $\Omega_{N}$ is defined so as to have the normalization
$$
\mathrm{Res}_{\infty_+}\Omega_N=1=-\mathrm{Res}_{\infty_-}\Omega_N
$$

\section{Separation of variables.}
The next step in our analysis of the model is to find a system of local Darboux coordinates on the symplectic manifold $M_N$. To do this, we apply Sklyanin's method of (classical) separation of variables, as explained in \cite{sklyaninsov}.  

\bigskip

From the symmetries (\ref{identities}) of $\mT(z)$, we have that 
\begin{align}\label{mores}A(z^{-1})=-D(z), \ \ C(z^{-1})=C(z)\end{align}
\begin{align}
C(-z)=-C(z), \ \ A(-z)=A(z)
\end{align}
In view of the symmetries of $C(z)$, it is natural to consider
$$
\tilde{C}(z)=\frac{C(z)}{z+z^{-1}}
$$
which satisfies 
$$
\tilde C(z^{-1})=\tilde{C}(z), \ \ \tilde{C}(-z)=\tilde{C}(z)
$$
and is therefore a function of $\lambda$.  In fact, $\tilde{C}(\lambda)$ is a polynomial of degree $N-1$, and following Sklyanin \cite{sklyaninsov}, we may introduce coordinates $(\lambda_1,\ldots,\lambda_{N-1},Q)$ as its zeros and asymptotic as $\lambda\rightarrow\infty$:   
\begin{align}
\label{lambdacoords}
\tilde{C}(z)=Q\prod_{k=1}^{N-1}(\lambda-\lambda_k)
\end{align} 
Note that in order to obtain a well defined set of coordinates in this fashion one must specify a locally consistent ordering of the roots $\lambda_j$. However, the angle coordinates constructed in Section 5 will turn out to be independent of this choice of ordering.    Note also that the leading coefficient $Q$ is given by
$$
Q=\sum_{j=1}^Nf_j\left((k_j/a_j)\prod_{i>j}k_i^2\xi_- -(k_j/a_j)^{-1}\prod_{i>j}k_i^{-2}\xi_-^{-1}\right)
$$
We also introduce the corresponding multi-valued $w$-coordinates $$w_j+w_j^{-1}=\lambda_j$$  Observe that since when $C(z)$ vanishes the reflection monodromy matrix becomes upper triangular, the points $(w,\zeta)=(w_j, A(z_j^{\pm1}))$ where $z_j^2=w_j$ lie on the curve $\Sigma$, and the points $(\lambda,\zeta)=(\lambda_j, A(z_j^{\pm1}))$ lie on the curve $\Gamma$. \bigskip

Let us fix a particular branch of the equation $w+w^{-1}=\lambda $ to give us a locally defined set of functions $w_1,\ldots,w_{N-1}$. Again, the angle coordinates we construct will be independent of this choice.   We may then introduce a further $(N-1)$ local coordinates
\begin{align}
\label{zetacoords}
\zeta_k=A(w_k)
\end{align} 
In terms of the function $y$ defined by (\ref{ycoord}), we have
$$
y_j:=\zeta_j-t(z_j)=\frac{A(z_j)-D(z_j)}{2}
$$
which by (\ref{mores}) is independent of our choice of branch of $w$.  

\bigskip

We now have the following proposition, which is proved by direct calculation using formulae (\ref{exppb}) for the Poisson brackets of reflection monodromy matrix elements.  
\begin{prop}
The coordinates $(Q,w_1,\ldots,w_{N-1}\ ;\ P,\zeta_1,\ldots,\zeta_{N-1})$ are log-canonical: we have
\begin{align}
\label{logcanonical}
\{w_k,\zeta_j\}&=2\delta_{j,k}w_j\zeta_k, \ \ \{Q,P\}=2QP
\end{align} 
and the Poisson brackets of all other pairs of coordinates are zero.  
\end{prop}

To summarize, we obtain a system of log-canonical coordinates consisting of the asymptotics $Q,P$ of $\tilde{C}(\lambda),A(z)$ respectively, together with a degree $(N-1)$ divisor $(w,\zeta)=(w_k,\zeta_k)$ on $\Sigma$ which projects onto the zero locus of the polynomial $\tilde{C}(\lambda)$.
%
%
%
%

\section{Linearization of flows and algebraic integrability.}
In this section we explain how to construct affine coordinates on the Liouville tori in $M_N$ cut out by the reflection Hamiltonians $\{P_j\}$, with respect to which the Hamiltonian flows of the $P_j$ correspond to linear motion with constant velocity.  To do this, we will use the Hamilton-Jacobi method; for further details, see \cite{arnold},\cite{babelon},\cite{harnad} and references therein.  

\bigskip

The first step is to use the log-canonical coordinates constructed in the previous section to write down a local expression for a primitive $\alpha$ for the symplectic form on $M_N$. We find
\begin{align}
\alpha&=\frac{\log P}{2Q}dQ+\frac{1}{2}\sum_k\log(\zeta_k)\frac{dw_k}{w_k}
\end{align}
We must now restrict $\alpha$ to the level sets of the reflection Hamiltonians $P_j$ and integrate in order to form the Hamilton-Jacobi action. The final step consists of differentiating with respect to the invariants $P_j$ to obtain the canonically conjugate angle variables $F_j$.  The action is given by
$$
S(Q,\lambda_1,\cdots,\lambda_{N-1},P_1,\ldots,P_N)=\frac{(\log P)(\log Q)}{2}+\frac{1}{2}\sum_{k=1}^{N-1}\int_{w_0}^{w_k}\log(\zeta)\frac{dw}{w}
$$
where the integral is understood as being taken on the spectral curve $\Sigma$.  We therefore find
\begin{align}
\label{abelianintegrals}
F_j&=\frac{\partial S}{\partial P_j} = \begin{dcases} \sum_{k=1}^{N-1}\int_{\lambda_0}^{\lambda_k}\Omega^+_j &   \ 1\leq j\leq N-1  \\ 
\frac{\log(Q)}{4(P+P^{-1})}- \frac{1}{4(P+P^{-1})}\sum_{k=1}^{N-1}\int_{\lambda_0}^{\lambda_k}\Omega_{N}& j=N\end{dcases}
\end{align}
where we may now regard the integrals as being taken on the genus $N-1$ curve $\Gamma$.  The symplectic form being written as 
$$
\omega=\sum_{k=1}^NdF_k\wedge dP_k
$$
the time evolution under the reflection flows becomes linear in these coordinates:
$$
F_j(t_k)=F_j(0)+t_k\delta_{jk}
$$
 Note that the coordinates $F_1,\ldots,F_{N-1} $ coincide with the Abel map applied to the degree $g(\Gamma)=N-1$ divisor 
$$
\mathcal{D}=p_1+\cdots+p_{N-1}
$$
on $\Gamma$, where we write $p_j$ for the point $(\lambda,y)=(\lambda_j,y_j)$.  Hence the reflection flow is linearized on the $\mathrm{Jac}(\Gamma)$, the Jacobian variety of $\Gamma$, which establishes the algebraic integrability of the system.

\section{Action-angle variables.  }
In this section we explain how to construct complex action-angle variables for the system.  Let us choose a canonical basis $$(A_1,\ldots, A_{N-1},B_1,\ldots,B_{N-1})$$ for $H_1(\Gamma_0,\bZ)$, where $\Gamma_0$ is some fixed spectral curve.  By Gauss-Manin, this choice of basis has a well defined propagation to a canonical homology basis for all nearby spectral curves $\Gamma$.    We will also need to introduce $\gamma:=A_N$ ,  a contractible loop on $\Sigma$ winding once around the point $\infty_+$. 

\bigskip

In order to define the action-angle variables, we must choose a lifting of $A_1,\ldots,A_N$ to homology classes $\tilde{A}_1,\ldots,\tilde{A}_N$ on the curve $\Sigma^c$ obtained by deleting slits between the branch points of the multi-valued function $\log(\zeta)$ on $\Sigma$.  On the cut Riemann surface $\Sigma^c$, we have a well-defined meromorphic differential
\begin{align}
\eta=\log(\zeta)\frac{dw}{w}
\end{align}
Then the action variables $J_1,\ldots,J_N$ are defined as the $A$-periods of the differential $\eta$:
\begin{align}
\label{actions}
J_k=\oint_{A_k}\eta, \ \ 1\leq k\leq N
\end{align}
A priori, this definition of the action variables depends on our choice of lifting of the homology classes $A_i$.  However, the following proposition shows that this dependence is of a tame nature.  

\begin{prop}
Let $\{J_k\}, \{J_k'\}$ be two sets of coordinates defined by formula (\ref{actions}) for two different choices of sets of lifts $\{\tilde{A}_k\},\{\tilde{A}_k'\}$ of the homology classes $\{[A_k]\}\subset H_1(\Gamma,\bZ)$, having the same winding numbers around $\infty_+$.  Then each difference $J_k-J_k'$ is a constant function on $M$, and the map 
\begin{align}
\label{coordchange}
(P_1,\ldots,P_N)\longmapsto (J_1,\ldots,J_N)
\end{align}
is a change of coordinates.  
\end{prop}
\begin{proof}
Let us first show that (\ref{coordchange}) is a change of coordinates.  For this, note that 
\begin{align}
\frac{\partial J_i}{\partial P_k}=\oint_{A_i}\Omega_k
\end{align}
Since the pairing between $H^0(\Gamma,K)$ and the span of the $A$-cycles is perfect, and $\Omega_N$ is the only differential of the $\Omega_j$ with nonzero residue at $\infty_+$, it follows that the Jacobian matrix of (\ref{coordchange}) is of full rank, which shows that $(\ref{coordchange})$ is a change of coordinates.  Now to prove the first assertion of proposition amounts to showing that 
$$\{F_k,(J_i-J_i')\}=\frac{\partial}{\partial P_k}(J_i-J_i')=0$$
for all $j,k$.  But  since the differentials $\Omega_1,\ldots,\Omega_{N-1}$ are well defined meromorphic differentials on $\Gamma$, and by definition $A_N,A'_N$ have the same winding number around $\infty_+$, we have
$$
\frac{\partial J_i}{\partial P_k}=\oint_{A_i}\Omega_k=\oint_{A_i'}\Omega_k=\frac{\partial J_i'}{\partial P_k}
$$
\end{proof}
Let us also note that , again up to a shift by an additive constant, the action variable $J_N$ is given by $J_N=2\pi i\log P$.   

\bigskip
With these results in hand we can proceed to the construction of the angle variables $\omega_k$ as the coordinates canonically conjugate to the $J_k$ by the Hamilton-Jacobi method: 
\begin{align}
\omega_k=\frac{\partial S}{\partial J_k}
\end{align}
Note that these coordinates are independent of our choices of representative for the homology classes $A_k$, and the differentials $d\omega_k$ are $\tau$-invariant and thus descend to the curve $\Gamma$.  
Moreover, for $1\leq k\leq N-1$ we have
\begin{align*}
\oint_{A_i}d\omega_k&=\frac{\partial}{\partial J_k}\oint_{A_i}{dS}\\
&=\frac{\partial}{\partial J_k}\oint_{A_i}{\left(\sum_r p_rdq_r+F_rdP_r\right)}\\
&=\frac{\partial}{\partial J_k}\oint_{A_i}{\alpha}\\
&=\delta_{ik}
\end{align*}
which shows that the angle variables are indeed normalized correctly with respect to the $A$-cycles of $\Gamma$,  and that all $A$-periods of the differential $d\omega_N$ vanish. Note that if $p\in\{\infty_\pm\}$, and $\gamma_{p}$ is a contractible loop in $\Gamma$ with winding number 1 around $p$, we also have
 \begin{align*}
\oint_{\gamma_p}d\omega_k&=\frac{\partial}{\partial J_k}\oint_{\gamma_p}{dS}\\
&=\frac{\partial}{\partial J_k}\oint_{\gamma_p}{\eta}\\
&=\pm\delta_{kN}
\end{align*}
which shows that the differentials $d\omega_1,\ldots,d\omega_{N-1}$ are holomorphic, and that
\begin{align*}
\mathrm{Res}_{\infty_+}d\omega_N=1=-\mathrm{Res}_{\infty_-}d\omega_N.
\end{align*}

\section{Solutions in theta functions.}
We will now apply the geometric description of the system given in the previous sections to write explicit formulas for the flows of the reflection Hamiltonians using Riemann theta functions. 

\bigskip

 Let $(A_i,B_i)$ be the canonical homology basis and $\{dw_j\}$ be the normalized abelian differentials constructed in the previous section. The {\em matrix of } $b$-{\em periods} corresponding to this data is the $(N-1)\times (N-1)$ symmetric matrix 
 
$$
\mathcal{B}_{jk}=\oint_{B_j}d\omega_k, \ \ 1\leq j,j\leq N-1
$$
 This matrix gives the rise to the model
$$
\mathrm{Jac}(\Gamma)=\bC^{N-1}/(\bZ^{N-1}+\mathcal{B}\bZ^{N-1})
$$
for the Jacobian of $\Gamma$.
Expanding $d\omega_j=\sum_k \mathcal{N}_{jk}\Omega_k$ where $\mathcal{N}_{jk}\in\bC$, we define the normalized angle variables $$\widetilde{F}_j=\sum_k \mathcal{N}_{jk}F_k, \ \ j=1,\ldots, N-1$$
$$\widetilde{F}_N=4(P+P^{-1})\left(\sum_{k=1}^{N} \mathcal{N}_{jk}F_k\right)$$ Note that $\mathcal{N}_{jN}=0$ for $j=1,\ldots,N-1$ and $\mathcal{N}_{NN}=1$ so that we have
$$
\widetilde{F}_N={\log Q}-\sum_{k=1}^{N-1}\int_{\lambda_0}^{\lambda_k}d\omega_{N}
$$ 
In these coordinates the time evolution takes the form
\begin{align}
\label{normalizedjac}
\widetilde{F}_i(t_k)=\widetilde{F}_i(0)+t_k\mathcal{N}_{ik}, \ \ i=1,\ldots, N-1
\end{align}
\begin{align}
\label{normalizedfn}
\widetilde{F}_N(t_k)=\widetilde{F}_N(0)+c_kt_k
\end{align}
where $c_k=4(P+P^{-1})\mathcal{N}_{Nk}$.  
If we define the normalized Abel map with base point $p_0$
\begin{align}
\label{basedabel}
\mathcal{A}_j(p_1+\ldots+p_{N-1})=\sum_{k=1}^{N-1}\int_{p_0}^{p_k}d\omega_j
\end{align}
we have
$$
\mathcal{A}(\mathcal{D}(t))=\mathcal{A}(\mathcal{D}(0))+t_kU^{(k)}
$$
where the velocity vector $U^{(k)}$ is given by
$$
U^{(k)}_j=\mathcal{N}_{jk}
$$
and
$$
\widetilde{F}_N(t)=\widetilde{F}_N(0)+c_kt_k.
$$
Let us now recall some background on theta functions.  For a more detailed discussion of this subject, see \cite{belokolos} and references therein.  The Riemann theta function associated to the spectral curve $\Gamma$ and its matrix of $b$-periods $\mathcal{B}$ is the following holomorphic function on $\bC^{N-1}$:
\begin{align}
\label{theta}
\theta(z)=\sum_{m\in\bZ^{N-1}}e^{2\pi i (m,z) +\pi i (\mathcal{B}m,m)}
\end{align}
The theta function is automorphic with respect to the lattice of periods of $\Gamma$: if $n\in\bZ^{N-1}$, we have
\begin{align}
\nonumber\theta(z+n)&=\theta(z)\\
\label{automorphy}\theta(z+\mathcal{B}n)&=\exp\left(-2\pi i(n,z)-\pi i (\mathcal{B}n,n) \right)\theta(z)
\end{align}
From these formulas, it follows that the divisor $\Theta$ of $\theta(z)$ is a well defined analytic subset of the Jacobian $\mathrm{Jac}(\Gamma)$.  
Let us fix a so-called {\em odd non-singular} point $e\in\Theta\subset\bC^{N-1}$ of the theta divisor.  Then the third kind differential $\widetilde{\Omega}_N$ can be expressed in terms of the odd theta function $\theta_e(z):=\theta(z+e)$ as
$$
\widetilde{\Omega}_N(p)=d\log\left(\frac{\theta_e(\int^{p}_{\infty_+}\omega)}{\theta_e(\int^{p}_{\infty_-}\omega)} \right)
$$ 
where we use the shorthand notation 
$$
\int^{p}_{q}\omega=\mA(p)-\mA(q)
$$
Hence from our formula (\ref{normalizedfn}) for the time evolution of $\widetilde{F}_N$, we obtain the following expression for the time evolution of the observable $Q$ under the Hamiltonian flow of the reflection Hamiltonian $P_k$:
\begin{align}
\label{Qevolution}
Q(t_k)=Q(0)e^{c_kt_k}\prod_{j=1}^{N-1}\frac{\te(\int_{p_j(t_k)}^{\infty_+}\omega)\te(\int_{p_j(0)}^{\infty_-}\omega)}{\te(\int_{p_j(0)}^{\infty_+}\omega)\te(\int_{p_j(t_k)}^{\infty_-}\omega)}
\end{align}
However this formula is of limited practical value, in that it requires knowledge of the points $p_k(t)$ for all times $t$, whereas all we know explicitly is the (linear) time evolution of $\mA(\mathcal{D}(t))$. We may remedy this defect as follows.  Let $K$ denote the Riemann point for the based Abel map (\ref{basedabel}).  Consider two non-special effective degree $g=N-1$ divisors
$$
\mathcal{D}=p_1+\cdots+p_g, \ \ \mathcal{D'}=q_1+\cdots+q_g
$$
and form the meromorphic function
$$
m(p)=\prod^g_{j=1}\frac{\te(\int^p_{p_j}\omega)}{\te(\int^p_{q_j}\omega)}\cdot\frac{\theta(\mA(p)-\mA(D')-K)}{\theta(\mA(p)-\mA(D)-K)}
$$
which must be constant since it has no poles.  We therefore obtain, for any point $q$ on the curve,
$$
\prod^g_{j=1}\frac{\te(\int^p_{p_j}\omega)\te(\int^q_{q_j}\omega)}{\te(\int^p_{q_j}\omega)\te(\int^q_{p_j}\omega)}=\frac{\theta(\mA(p)-\mA(D)-K)}{\theta(\mA(p)-\mA(D')-K)}\frac{\theta(\mA(q)-\mA(D')-K)}{\theta(\mA(q)-\mA(D)-K)}
$$
Applying this formula to in the case $\mathcal{D}=p_1(t)+\cdots+p_g(t), \mathcal{D}'=p_1(0)+\cdots+p_g(0)$, $p=\infty_+, q=\infty_-$, we find
\begin{align*}
Q(t_k)&=Q(0)e^{c_kt_k}\frac{\theta(\mA(\infty_+)-\mA(\mathcal{D}(t))-K)}{\theta(\mA(\infty_+)-\mA(\mathcal{D}(0))-K)}\frac{\theta(\mA(\infty_-)-\mA(\mathcal{D}(0))-K)}{\theta(\mA(\infty_-)-\mA(\mathcal{D}(t))-K)}\\
&=Q(0)e^{c_kt_k}\frac{\theta(\mA(\infty_+)-\mathcal{A}(\mathcal{D}(0))-t_kU^{(k)}-K)\theta(\mA(\infty_-)-\mA(\mathcal{D}(0))-K)}{\theta(\mA(\infty_-)-\mathcal{A}(\mathcal{D}(0))-t_kU^{(k)}-K)\theta(\mA(\infty_+)-\mA(\mathcal{D}(0))-K)}
\end{align*}
which is an explicit formula for the time evolution of $Q$.  

\bigskip

We now turn to the problem of reconstructing the full reflection monodromy matrix.  For this, we introduce the following meromorphic function on $\Gamma$:
\begin{align}
\rho&=\frac{Q(z+z^{-1})}{P+P^{-1}}\cdot\frac{y+h(\lambda)}{C(z)}\\
&=\frac{1}{P+P^{-1}}\cdot\frac{y+h(\lambda)}{(\lambda+2)\prod_{k=1}^{N-1}(\lambda-\lambda_k)}
\end{align}
where we write $$h(\lambda)=\frac{A(z)-D(z)}{2}.$$ 
The relevance of the function $\rho$ to our problem is that the vector
$$
\psi=\left(1,\frac{(P+P^{-1})}{Q(z+z^{-1})}\cdot\rho\right)^t
$$
spans the eigenspace of $\mT(z)$ corresponding to the given point on the spectral curve.  We have the following proposition characterizing the function $\rho$.   
\begin{prop}
The meromorphic function $\rho$ has exactly $N$ poles, $N-1$ of them at the divisor $\mathcal{D}$, and one at the point $q^+=(-2,h(-2))$ lying over $\lambda=-2$.  In addition, $\rho$ has a zero at $\infty_-$.  Its value at $\infty_+$ is 
$$
\rho(\infty_+)=1
$$
\end{prop}
\begin{proof}
The assertion about the pole at $q^+$ follows from the identity $$h^2(-2)=\mathcal{Q}_{2N}(-2)=\left(\frac{\xi-\xi^{-1}}{4}\right)\prod_{k=1}^N(\omega_k+a_k^2+a_k^{-2})^2
$$
\end{proof}
By the Riemann-Roch theorem there is generically a unique such $\rho$, which can be written as
$$
{\rho}=\frac{{\te(\int^p_{\infty_-}\omega} )\theta(\mA(p)-\mA(D)-K+W)}{{\te(\int^p_{q^+}\omega} )\theta(\mA(p)-\mA(D)-K)}\cdot\frac{{\te(\int^{\infty_+}_{q^+}\omega} )\theta(\mA(\infty_+)-\mA(D)-K)}{{\te(\int^{\infty_+}_{\infty_-}\omega} )\theta(\mA(\infty_+)-\mA(D)-K+W)}
$$
where the vector $W$ is defined as the vector of $b$-periods of the unique normalized third kind differential $\Omega_{\infty_-,q^+}$  with residue 1 at $\infty_-$ and residue $-1$ at $q^+$:
$$
W_j=\oint_{B_j}\Omega_{\infty_-,q^+}
$$
Hence the time evolution of $\rho$ under the flow of the reflection Hamiltonian $P_k$ is given by the explicit formula
\begin{align*}
\rho(p,t_k)=&\frac{{\te(\int^p_{\infty_-}\omega} )\theta(\mA(p)-\mA(D(0))-t_kU^{(k)}-K+W)}{{\te(\int^p_{q^+}\omega} )\theta(\mA(p)-\mA(D(0))-t_kU^{(k)}-K)}\\
&\cdot\frac{{\te(\int^{\infty_+}_{q^+}\omega} )\theta(\mA(\infty_+)-\mA(D(0))-t_kU^{(k)}-K)}{{\te(\int^{\infty_+}_{\infty_-}\omega} )\theta(\mA(\infty_+)-\mA(D(0))-t_kU^{(k)}-K+W)}
\end{align*}
From this we can reconstruct the eigenvector $\psi(p)$ and therefore the full reflection monodromy matrix $\mT(z)$.  

\end{document}